\documentclass{article}
\usepackage[english]{babel}
\usepackage{amsmath}
\usepackage{amsfonts}
\usepackage{graphicx}
\usepackage[colorlinks=true, allcolors=blue]{hyperref}
\usepackage{comment}
\usepackage{enumitem}
\usepackage{mathtools}
\usepackage{amssymb}
\usepackage{amsthm}
\usepackage{subcaption}

\usepackage {tikz}
\usetikzlibrary {positioning}
\tikzset{main node/.style={circle,fill=blue!40,draw,minimum size=1.5cm,inner sep=0pt},}

\newtheorem{prop}{Proposition}
\newtheorem{rem}{Remark}
\newtheorem{lem}{Lemma}
\newtheorem{coro}{Corollary}
\newtheorem{theo}{Theorem}

\newcommand*{\prob}{\mathbb{P}}
\newcommand{\vect}[1]{\boldsymbol{#1}}

\DeclareMathOperator\supp{supp}

\title{Jarzyski's equality and Crooks' fluctuation theorem for general Markov chains with application to decision-making systems}
\author{Pedro Hack, Sebastian Gottwald, Daniel A. Braun}
\date{ }

\begin{document}
\maketitle

\begin{abstract}
We define common thermodynamic concepts purely within the framework of general Markov chains and derive Jarzynski’s equality and Crooks’
fluctuation theorem in this setup. In particular, we regard the discrete time case that leads to an asymmetry in the definition of work that appears in the usual formulation of Crooks’ fluctuation
theorem. We show how this asymmetry can be avoided with an additional condition regarding the energy protocol. The general formulation in terms of Markov chains allows transferring the results to other application areas outside of physics. Here, we discuss how this framework can be applied in the context of decision-making. This involves the definition of the relevant quantities, the assumptions that need to be made for the different fluctuation theorems to hold, as well as the consideration
of discrete trajectories instead of the continuous trajectories, which are relevant in physics.

%We discuss the assumptions needed to derive both Jarzynski's equality and Crooks' fluctuation theorem when the dynamics are Markovian. In particular, we provide a new derivation of the former and an additional condition which is needed to prove the latter. We notice, because of the way in which one can prepare a thermodynamic system in the Boltzmann distribution, the new requirement is fulfilled by all previous experimental setups supporting Crooks' fluctuation theorem.
\end{abstract}

\section{Introduction}

Over the last 20 years, several advances in thermodynamics have led to the development of results relating equilibrium quantities to non-equilibrium trajectories. Those advances have crystallized in a new area of research, \emph{non-equilibrium thermodynamics}, where these relations play a major role \cite{seifert2012stochastic,jarzynski2011equalities,jarzynskia2008nonequilibrium}. Among them, two of the most remarkable, are Jarzynski's equality \cite{jarzynski1997equilibrium,jarzynski2004nonequilibrium,jarzynski1997nonequilibrium} and Crooks' fluctuation theorem \cite{crooks1999entropy,crooks1998nonequilibrium}, for which experimental evidence has been reported in several contexts: unfolding and refolding processes involving RNA \cite{collin2005verification,liphardt2002equilibrium}, electronic transitions between electrodes manipulating a charge parameter \cite{saira2012test}, rotation of a macroscopic object inside a fluid surrounded by magnets where the current of a wire attached to the macroscopic object is manipulated \cite{douarche2005experimental}, and a trapped-ion system \cite{an2015experimental,smith2018verification}.

These two results have been derived under several assumptions in the context of non-equilibrium thermodynamics, including both deterministic \cite{jarzynski1997nonequilibrium,jarzynski2004nonequilibrium} and stochastic dynamics \cite{jarzynski1997equilibrium,crooks1998nonequilibrium,crooks1999entropy,crooks2000path}. Moreover, it has been argued that both results can be obtained as a consequence of Bayesian retrodiction in a physical context \cite{Buscemi2021}. Here, we derive both of them using only concepts from the theory of Markov chains. This allows us to both distinguish the mathematical from the physical assumptions underlying them and, thus, to make them available for application in other areas where the framework of thermodynamics may be useful. The distinction between mathematical and physical assumptions will be of particular importance for the definition of work, as we will see, since the usual definition based on physical considerations leads to an asymmetry of the definition in processes that run either forward or backward in time---see Figure~\ref{crooks def} for a simple example. This is relevant, for instance, when analysing trajectories in terms of their work value, if we do not know whether they were recorded in the forward direction or whether they have been generated by playing them backwards. Ideally, we would like to be able to ascribe work values directly to trajectories without any additional information.

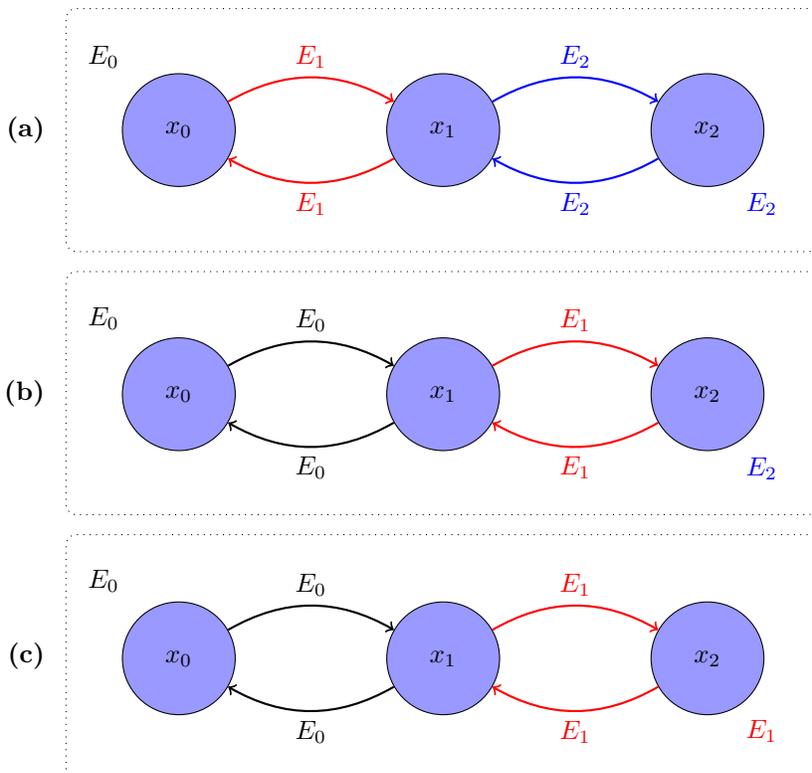
\begin{figure}[!tb]
	\centering
	\begin{tikzpicture}[auto,node distance=2.8cm]
		\node[main node] (1) [xshift=1.5cm] {$x_0$};
		\node[main node] (2) [right = 2cm of 1]  {$x_1$};
		\node[main node] (3) [right = 2cm of 2] {$x_2$};
		%\node[main node] (4) [right = 2cm of 3]  {$x_1$};
		\node[main node] (5) [below=2cm of 1]  {$x_0$};
		\node[main node] (6) [right = 2cm of 5]  {$x_1$};
		\node[main node] (7) [right = 2cm of 6]  {$x_2$};
		\node[main node] (8) [below=2cm of 5]  {$x_0$};
		\node[main node] (9) [right = 2cm of 8]  {$x_1$};
		\node[main node] (10) [right = 2cm of 9]  {$x_2$};
		%\node[main node] (6) [below right= 2.5cm of 2]  {$x_1$};
		\node[rounded corners, draw, dotted, text height = 3cm,minimum  width=10cm, xshift=5cm, label={[anchor=west,left=1.5mm]180: \textbf{(a)}}]  (main) {};
		\node[rounded corners, draw,dotted, text height = 3cm,minimum  width=10cm, xshift=5cm, yshift=-3.5cm, label={[anchor=west,left=1.5mm]180: \textbf{(b)}}]  (main) {};
		\node[rounded corners, draw,dotted, text height = 3cm,minimum  width=10cm, xshift=5cm, yshift=-7cm, label={[anchor=west,left=1.5mm]180: \textbf{(c)}}]  (main) {};
		%\node[main node] (3) [right = 2cm of 2] {$x_2$};
		
		\node at (0.5,0.98) {$E_0$};
		\node at (0.5,-2.5) {$E_0$};
		\node at (0.5,-6) {$E_0$};
		\node [blue] at (9.25,-0.98) {$E_2$};
		\node [blue] at (9.25,-4.5) {$E_2$};
		\node [red] at (9.25,-8) {$E_1$};
		%\node at (0.25,-2.62) {$E_0$};
		%\node at (9.5,-4.57) {$E_1$};
		
		\path[draw,thick,->]
		(1) edge [bend left,red]  node {$E_1$} (2)
		(2) edge [bend left,blue]  node {$E_2$} (3)
		(3) edge [bend left,blue]  node {$E_2$} (2)
		(2) edge [bend left,red]  node {$E_1$} (1)
		(5) edge [bend left]  node {$E_0$} (6)
		(6) edge [bend left,red]  node {$E_1$} (7)
		(6) edge [bend left]  node {$E_0$} (5)
		(7) edge [bend left,red]  node {$E_1$} (6)
		(8) edge [bend left]  node {$E_0$} (9)
		(9) edge [bend left,red]  node {$E_1$} (10)
		(10) edge [bend left,red]  node {$E_1$} (9)
		(9) edge [bend left]  node {$E_0$} (8)
		;
		%(3) edge [bend left]  node {$E_2$} (2);
		%(3) edge node {} (1);
	\end{tikzpicture}
	\caption{Relation between a forward process, its corresponding backward process and the definition of work. We consider a trajectory $\vect{x} =(x_0,x_1,x_2)$ and three energy functions $E_0,E_1,E_2$. The upper (bottom) line of arrows represents the forward (backward) process. A work step $W_i = E_i(x_{i-1})-E_{i-1}(x_{i-1})$ is typically defined as the change in energy due to the external change of the energy function, whereas a heat step $Q_i = E_i(x_i) - E_i(x_{i-1})$ is defined as the change in energy due to internal state changes. \textbf{(a)} Typical relation between the forward and backward processes in physics \cite{crooks1998nonequilibrium}. Work in the forward process would be $W_F = E_1(x_0) - E_0(x_0) + E_2(x_1)-E_1(x_1)$, whereas the backward work under the same definition would be $W_B = E_1(x_1)- E_2(x_1)$. 		
		Instead backward work is usually defined as $W_B=E_1(x_1)- E_2(x_1)+E_0(x_0)-E_1(x_0)$ to fulfil the physical time-reversal symmetry $W_F=-W_B$. \textbf{(b)} Another typical protocol in physics \cite{crooks2000path,crooks1999excursions}. In this case the asymmetry is the other way round, where $E_2$ does not influence the forward process. \textbf{(c)} Symmetric protocol where both forward and backward work follow the same definition with $W_F = E_1(x_1)-E_0(x_0)$ and $W_B = E_0(x_1)-E_1(x_1)= -W_F$. This is the protocol we propose in Section~\ref{sec:crooks}.}
	\label{crooks def}
\end{figure}

One of the application areas where the framework of thermodynamics has recently  been investigated outside the realm of physics is the analysis of simple learning systems \cite{goldt2017stochastic,perunov2016statistical,england2015dissipative,still2012thermodynamics,seifert2012stochastic,ortega2013thermodynamics,parr2020markov,da2021bayesian,GottwaldBraun2020,Boyd2022}, and in particular the problem of decision-making under uncertainty and resource constraints \cite{wolpert2006information,tishby2011information,ortega2011information,ortega2013thermodynamics,genewein2015bounded,wolpert2016free}. The basic analogy follows from the idea that decision-making involves two opposing \emph{forces}: (i) the tendency of the decision-maker towards better options (equivalently, to maximize a function called \emph{utility}) and (ii) the restrictions on this tendency given by the limited information-processing capabilities of the decision-maker, which prevents it from always picking the best option and is usually modelled by a bound on the entropy of the probability distribution that describes the decision-maker's behaviour.Thermodynamic systems are also explained in terms of two opposing forces. The first being the energy, which the system tries to minimize, and the second being entropy, which prevents the minimization of energy to its full extent. Thus, in both cases we formally deal with optimization problems under information constraints and, thus, we can conceptualize both decision-making and thermodynamics in terms of information theory. In particular, we can consider the environment in which a decision is being made or a thermodynamic system is immersed as a source of information  (in the form of either utility or energy) which, due to the noise (modelled by entropy), reaches the decision-making or thermodynamic system with some error. This results in an imperfect response by the system.
%due to the noise, do not reach

The analogy between thermodynamics and decision-making is not restricted to the equilibrium case \cite{wolpert2006information,tishby2011information,ortega2011information,ortega2013thermodynamics,genewein2015bounded,wolpert2016free}, but can be taken further from equilibrium to non-equilibrium systems. In particular, the aforementioned fluctuation theorems of Jarzynski and Crooks have been previously suggested to apply to decision-makers that adapt to changing environments \cite{grau2018non}. In this previous work, hysteresis and adaptation have been investigated in decision-makers, however, based on the physical convention of defining work differently for forward and backward processes. Here we improve on the work there by replacing this convention with a different energy protocol that naturally entails a symmetric definition of work and by weakening the assumptions that are actually needed in order for the fluctuation theorems to hold in the context of general Markov chains. Given the fact that the literature on this topic belongs for the most part to thermodynamics, we adopt the thermodynamic notation here. In particular, we consider energy functions instead of utilities and take Markov chains as starting point. 

%Here, we take a Markov chain approach to Jarzynski's equality and Crooks' fluctuation theorem.
Our manuscript is organized as follows. In Section \ref{basics}, we introduce notions of work and other thermodynamic concepts that are inherent to Markov chains, that is, in contrast to the formalism in physics \cite{crooks1998nonequilibrium,crooks1999entropy,crooks2000path}, we start with the assumption of a Markov chain and deduce all other concepts from that without the need to presuppose the existence of an external energy function. We discuss under what conditions these concepts are uniquely specified from the Markov chain. In Section \ref{sec:jarz}, we use this framework to weaken the derivation of Jarzynski's equality (Theorem \ref{Jarzynski's equality}) in the context of decision-making that was presented in \cite{grau2018non}. In Section \ref{sec:crooks}, we prove Crooks' fluctuation theorem (Theorem \ref{crooks fluctuation thm} and Corollary \ref{physics crooks}) within the same setup. 
In particular, we use an additional assumption which is not mandatory in Crooks' work \cite{crooks1998nonequilibrium,crooks2000path,crooks1999entropy}, but here is needed given the inherent nature of
our definition of work.
In fact, we provide an example in which the new requirement is violated and, as a consequence, Crooks' theorem is false everywhere  (Proposition \ref{counterexample}). In Section~\ref{sec:decision} we discuss how the concepts we have developed can be applied to decision-making systems.

%Fortunately, this new restriction is fulfilled by all previous experimental setups supporting Crooks' fluctuation theorem. 

Notice, for simplicity, we develop the results for discrete-time Markov chains with finite state spaces. However, the ideas can be translated, for example, to continuous state spaces by assuming for densities of Markov kernels the properties we use here for transition matrices. We include a few more details regarding this scenario in the discussion (Section \ref{sec:decision}), where we also briefly address the case of continuous-time Markov chains.
%in the discussion (Section \ref{sec:decision}).

\section{Thermodynamics for Markov chains}
\label{basics}

\subsection{Definitions of energy, heat, and work}
% or: \subsection{Energy, work, and heat of Markov chains}

In this section, we assume there is some stochastic process which can be modelled as a Markov chain and discuss the definition of energy, partition function, free energy, work, heat and dissipated work in such a context. We follow the terminology in \cite{levin2017markov}.

We call a finite number of random variables $\vect{X}=(X_n)_{n=0}^N$ over a finite state space $S$ a \emph{Markov chain} if we have for any $0<n \leq N$, that
\begin{equation*}
\prob(X_{n}=x_n |X_0 = x_0 ,\dots,X_{n-1} =x_{n-1}) \coloneqq \prob(X_{n}=x_n | X_{n-1} =x_{n-1})
\end{equation*}
for all $(x_0,x_1,\dots,x_n) \in S^{n+1}$.
Notice we can characterize $\vect{X}$ by a probability distribution $p_0$, where $p_0(x)=\prob(X_0=x)$, and $N$ \textit{transition matrices} $(M_n)_{n=1}^N$ given by 
\begin{equation*}
    (M_n)_{xy}\coloneqq \prob(X_n=x|X_{n-1}=y)
\end{equation*}
for all $x,y \in S$ and $1 \leq n \leq N$.
%In case there exist $n,m$, $1 \leq n,m \leq N$, such that $M_n \neq M_m$, then $\vect{X}$ is called \emph{inhomogeneous}. Otherwise, $\vect{X}$ is called \emph{homogeneous}.
Notice, any transition matrix $M_n$ is a \emph{stochastic matrix}, that is, we have $\sum_{x \in S} (M_n)_{xy}=1$ for all $y \in S$. 

If $p$ is a distribution, $M_n$ a transition matrix of $\vect{X}$ for some fixed    $n$, and we have $M_n p=p$, where $(M_n p)(x)\coloneqq \sum_{y \in S} (M_n)_{xy}p(y)$, then we say $p$ is a \emph{stationary distribution} of $M_n$. Note that, this terminology applies to a \emph{single} $M_n$, that is, a stationary distribution $p$ of $M_n$ is stationary with respect to the (homogeneous) Markovian dynamics of that fixed transition matrix, and generally not with respect to the (inhomogeneous) Markovian dynamics of $\vect{X}$. If $p$ fulfills
\begin{equation}
\label{det balance}
(M_n)_{yx} p(x) = (M_n)_{xy}  p(y) \qquad \forall x,y \in S, 
\end{equation}
then we say $p$ satisfies \emph{detailed balance} with respect to $M_n$. Note that such a $p$ is a stationary distribution of $M_n$. In case $M_n$ has a unique stationary distribution $p$ which satisfies detailed balance with respect to $M_n$, then we may simply say $M_n$ satisfies detailed balance. We say a transition matrix $M_n$ is \emph{irreducible} if for any pair $x,y \in S$ there exists an integer $m \geq 1$ such that 
$(M_n^m)_{xy}>0$.
% the probability of reaching $y$ starting at $x$ applying $M$ $m$ times is greater than zero.
Irreducible transition matrices have a useful property which we present in Lemma \ref{irred prop} (see \cite[Corollary 1.17 and Proposition 1.19]{levin2017markov} for a proof).

\begin{lem}
\label{irred prop}
If a transition matrix is irreducible, then it has a unique stationary distribution. Furthermore, the stationary distribution has non-zero entries.
\end{lem}

If $\vect{X}$ has initial distribution $p_0$, irreducible transition matrices $(M_n)_{n=1}^N$, and $p_N$ is the unique stationary distribution of $M_N$, then we say the Markov chain $\vect{Y} \coloneqq (Y_n)_{n=0}^N$ with initial distribution $p_N$ and transition matrices  $(M_{N-(n-1)})_{n=1}^N$ is the \emph{time reversal} of $\vect{X}$. Notice there are different notions of time reversal. A discussion can be found in \cite[Section III]{yang2020unified}.

Given a distribution $p$ on $S$ such that $p(x) > 0$ for all $x \in S$ and some $\beta >0$, we can associate to $p$ an \emph{energy function} $E$, that is, a function $E: S \to \mathbb{R}$ such that 
\[
p(x)=\frac{1}{Z}e^{-\beta E(x)} \qquad \forall x \in S \, ,
\]
where $Z \coloneqq \sum_{x \in S} e^{-\beta E(x)}$ is called a \emph{partition function} and $F \coloneqq -\frac{1}{\beta} \log (Z)$ the corresponding \emph{free energy}. Notice, given a distribution $p$ with two energy functions $E$ and $E'$, there is a constant $c \in \mathbb{R}$ such that we have
\begin{equation}
\label{ener shift}
E(x)=E'(x)+c \qquad \forall x \in S 
\end{equation} 
and, accordingly,
\begin{equation}
\label{free energy shift}
F = F'+c
\end{equation}
where $F$ and $F'$ are the free energies for $p$ using $E$ and $E'$, respectively. Thus, each distribution where all entries are strictly positive has, up to a constant, a unique energy function associated to it. 

For the remainder of this section, let $\vect{X} \coloneqq (X_n)_{n=0}^N$ be a Markov chain such that $p_0$ has non-zero entries and each transition matrix $M_n$ has a unique stationary distribution $p_n$ with non-zero entries. In this context, we call a family $\vect{E}=(E_n)_{n=0}^N$ of functions $E_n:S\to\mathbb{R}$ a \emph{family of energies} of $\vect{X}$, if $E_n$ is an energy function of $p_n$ for all $0 \leq n \leq N$. We define the \emph{work} of a realization $\vect{x}=(x_0,x_1,..,x_N) \in S^{N+1}$ of $\vect{X}$, with respect to a family of energies $\vect{E}$ of $\vect{X}$, as
\begin{equation}
\label{defi work}
    W_{\vect{X},\vect{E}}(\vect{x}) \coloneqq \sum_{n=0}^{N-1} E_{n+1}(x_n)-E_{n}(x_n) \, .
\end{equation}
Given another family of energies $\vect{E'}=(E'_n)_{n=0}^N$ of $\vect{X}$, we have 
\begin{equation}
\label{work preli}
 W_{\vect{X},\vect{E}}(\vect{x}) = W_{\vect{X},\vect{E'}}(\vect{x}) + (c_N -c_0) \, ,
\end{equation} 
where $c_n \coloneqq E_n-E'_n$ are constants by \eqref{ener shift}. Hence, without fixing a family of energies of $\vect{X}$, the work defined in \eqref{defi work}  is unique up to a constant. Whenever $\vect{X}$ and $\vect{E}$ are clear from the context, we may simply use $W$ instead of $W_{\vect{X},\vect{E}}$ for brevity.

Similarly, the \emph{heat} of a realization $\vect{x}$ of $\vect{X}$ with respect to a family of energies $\vect{E}$ is given by
\begin{equation}
    Q_{\vect{X},\vect{E}}(\vect{x}):= \sum_{n=1}^N E_n(x_n) - E_n(x_{n-1}) \, .
\end{equation}
Given another family of energies $\vect{E'}\coloneqq (E'_n)_{n=0}^N$ of $\vect{X}$, we have 
\begin{equation}
    Q_{\vect{X},\vect{E}}(\vect{x})
    =Q_{\vect{X},\vect{E'}}(\vect{x})
\end{equation}
by \eqref{ener shift}. We may, thus, use $Q_{\vect{X}}$ instead of $Q_{\vect{X},\vect{E}}$, or even $Q$ in case $\vect{X}$ is clear. 

Moreover, if $F_n$ is the free energy associated to $p_n$ for any $0\leq n\leq N$, we call $\Delta F_{\vect{X},\vect{E}} \coloneqq F_N-F_0$ the \emph{free energy difference} associated to $\vect{E}$, satisfying 
\begin{equation}
\label{free ener diff}
    \Delta F_{\vect{X},\vect{E}} =  \Delta F_{\vect{X},\vect{E'}}+ (c_N-c_0) 
\end{equation}
by \eqref{free energy shift}. While both  $W_{\vect{X},\vect{E}}$ and $\Delta F_{\vect{X},\vect{E}}$ depend on the difference between the constants $c_N$ and $c_0$, the so-called \emph{dissipated work} 
\[
W^d_{\vect{X},\vect{E}}(\vect{x}) \coloneqq W_{\vect{X},\vect{E}}(\vect{x}) - \Delta F_{\vect{X},\vect{E}}
\]
does not, that is, for any realization $\vect{x}$ of $\vect{X}$, we have
\begin{equation}
\label{dissip indep}
     W^d_{\vect{X},\vect{E}}(\vect{x}) = W^d_{\vect{X},\vect{E'}}(\vect{x})
\end{equation}
as a consequence of \eqref{work preli} and \eqref{free ener diff}. 

Notice, in the context of the Markov chain framework we have adopted here, the first law of thermodynamics is a direct consequence of the definitions of work and heat (see \eqref{1st law thermo} below),
\begin{equation*}
W_{\vect{X},\vect{E}}(\vect{x}) + Q_{\vect{X},\vect{E}}(\vect{x}) = E_N(x_N)-E_0(x_0).
\end{equation*}
The second law of thermodynamics can be obtained as well, which is a direct consequence of Jarzynski's equality, as we will see in Section \ref{sec:decision}.

\subsection{Main result: Fluctuation theorems for Markov chains}

Our main results are versions of \textit{Jarzynski's equality} \cite{jarzynski1997nonequilibrium,jarzynski2004nonequilibrium,jarzynski1997equilibrium} and \textit{Crooks' fluctuation theorem} \cite{crooks1998nonequilibrium,crooks1999entropy,crooks2000path} for Markov chains, the derivations of which can be found in Sections \ref{sec:jarz} and  \ref{sec:crooks} below. 

Consider a Markov chain $\vect{X}=(X_n)_{n=0}^N$ on a finite state space whose initial distribution $p_0$ has non-zero entries and whose transition matrices $(M_n)_{n=1}^N$ are irreducible. Then, for any family of energies $\vect{E} =(E_n)_{n=0}^N$ of $\vect{X}$, we have
\begin{equation*}
\big\langle e^{-\beta (W(\vect{X})-\Delta F)} \big\rangle = 1 \, ,
\end{equation*}
where $\langle \, \cdot \, 	\rangle$ denotes the expectation operator, and $W=W_{\vect{X},\vect{E}}$, $\Delta F= \Delta F_{\vect{X},\vect{E}}$. In physics, this result is known as Jarzynski's equality, which has been shown in the past to hold under various conditions. In Section \ref{sec:jarz}, we are giving a simple proof in the context of Markov chains as a direct consequence of the definitions of work and free energy (Theorem \ref{Jarzynski's equality}).

Moreover, if all transition matrices of $\vect{X}$ satisfy detailed balance, and $p_0$ is the stationary distribution of $M_1$, then for any possible work value $w$, we have
\begin{equation*}
    \frac{\prob(W_{\vect{X},\vect{E}}=w)}{ \prob(W_{\vect{Y},\vect{\widehat{E}}}=-w)} = e^{\beta(w- \Delta F)},
\end{equation*}
where $\vect{Y}$ is the time reversal of $\vect{X}$ with energies $\vect{\widehat{E}} = (\widehat{E}_n)_{n=0}^N$. The analogous result in physics is known as Crooks' fluctuation theorem. In Section \ref{sec:crooks}, we prove a slightly more general version in the context of Markov chains without detailed balance (Theorem \ref{crooks fluctuation thm}), which is then applied in Corollary \ref{physics crooks} to obtain Crooks' theorem.

\section{Jarzynski's equality for Markov chains}
\label{sec:jarz}

Jarzynki's equation was originally derived for deterministic dynamics \cite{jarzynski1997nonequilibrium,jarzynski2004nonequilibrium} (see also \cite{cohen2004note}) and later extended to stochastic dynamics \cite{jarzynski1997equilibrium} using a Master equation approach. Shortly after that, it was shown in the non-deterministic Markov chain context relying on assumptions about the time reversal of the dynamics \cite{crooks1998nonequilibrium}. In Theorem \ref{Jarzynski's equality} below, we see that in the context of Markov chains Jarzynski's equality is a straightforward consequence of the definitions of work and free energy. Importantly, it does not require any assumptions regarding time reversal, in contrast to the requirements in a previous decision-theoretic approach to fluctuation theorems in \cite{grau2018non}. 

For the proof of Jarzynski's equality for Markov chains, the basic observation is that we start with the expected value of a quantity closely related to the equilibrium distributions of our initial Markov chain $\vect{X}$, namely $e^{-\beta W(\vect{X})}$. With this in mind, we define a new Markov chain $\vect{Y}$ using the transition matrices of $\vect{X}$ and the equilibrium distributions of the individual steps. In particular, we define it in a way such that we cancel the dependency of $e^{-\beta W(\vect{X})}$ on $\vect{X}$ and end up with a constant, whose expected value over $\vect{Y}$ is the constant itself. We include the details in the following theorem.

% In Theorem \ref{Jarzynski's equality}, we develop a new derivation in the spirit of \cite{crooks1998nonequilibrium}, although without any assumption regarding time reversal.

\begin{theo}[Jarzynski's equality for Markov chains]
\label{Jarzynski's equality}
If $\vect{X}=(X_n)_{n=0}^N$ is a Markov chain on a finite state space $S$ whose initial distribution $p_0$ has non-zero entries and whose transition matrices $(M_n)_{n=1}^N$ are irreducible, then we have, for any family of energies  $\vect{E} =(E_n)_{n=0}^N$ of $\vect{X}$, 
\begin{equation}
\label{eq: jarz}
\big\langle e^{-\beta (W(\vect{X})-\Delta F)} \big\rangle = 1 \, ,
\end{equation}
where $W=W_{\vect{X},\vect{E}}$ and $\Delta F= \Delta F_{\vect{X},\vect{E}}$. 
%and $F_n \coloneqq -\frac{1}{\beta} \log (Z_n)$ for $0 \leq n \leq N$ where $Z_n \coloneqq \sum_{x \in S} e^{-\beta E_n(x)}$.
\end{theo}

\begin{comment}
\begin{prop}[Jarzynski's equality for Markov chains]
\label{Jarzynski's equality}
Consider a state space $S$, $|S|< \infty$, a sequence of functions $(E_n)_{n=0}^N$, $E_n:S \rightarrow \mathbb{R}$ and a sequence of distributions $(p_n)_{n=0}^N$, $p_n(x) \coloneqq \frac{1}{Z_n}e^{-\beta E_n(x)}$ where $Z_n \coloneqq \sum_{x \in S} e^{-\beta E_n(x)}$ for $0 \leq n \leq N$. Given a Markov chain $\vect{X}$ with initial distribution $p_0$ and whose transition matrices $(M_n)_{n=1}^N$ satisfy $M_n p_n = p_n$ for $1 \leq n \leq N$, then
\begin{equation}
\langle e^{-\beta W(\vect{x})} \rangle_{\prob(\vect{x})} = e^{-\beta \Delta F}
\end{equation}
where $\vect{x} \coloneqq (x_0,..,x_N) \in S^{N+1}$, $W=W_{\vect{X},\vect{E}}$, $F_n \coloneqq -\frac{1}{\beta} \log (Z_n)$ for $0 \leq n \leq N$ and $\Delta F= F_N-F_0$.
\end{prop}
\end{comment}

\begin{proof}
Notice, by Lemma \ref{irred prop}, we have $p_n(x) > 0$ for all $x \in S$ and $0 \leq n\leq N$. We first define a new Markov chain $\vect{Y} \coloneqq (Y_n)_{n=0}^N$ with initial distribution $p_N$ and with transition matrices \smash{$(\widehat{M}_n)_{n=1}^N$}, where for all $x,y \in S$,
\begin{equation}
\label{def rev}
(\widehat{M}_{n+1})_{xy} \coloneqq \frac{p_{N-n}(x)}{p_{N-n}(y)} ({M_{N-n}})_{yx}
\end{equation}
for $0 \leq n \leq N-1$. Notice $\widehat{M}_n$ is a stochastic matrix for $1 \leq n \leq N$ as we have for all $y \in S$
\begin{equation*}
\begin{split}
\sum_{x \in S} (\widehat{M}_n)_{xy} &= \sum_{x \in S} ({M_{N+1-n}})_{yx} \frac{p_{N+1-n}(x)}{p_{N+1-n}(y)} \\
&= \frac{1}{p_{N+1-n}(y)}\sum_{x \in S} ({M_{N+1-n}})_{yx} p_{N+1-n}(x)=1 ,
\end{split}
\end{equation*}
where we applied \eqref{def rev} in the first equality and the fact that $p_n$ is a stationary distribution for  $M_n$ for $1 \leq n \leq N$ by assumption in the last equality. Thus, $\vect{Y}$ is well defined. Note that by definition we have, 
\begin{equation*}
\begin{split}
   \prob(Y_{n+1}=x_{n+1} | Y_n =x_n) &= \prob(Y_{n+1}=x_{n+1} | Y_n =x_n,.., Y_0 = x_0) \\
   &= (\widehat{M}_{n+1})_{x_{n+1} x_n} 
   \end{split}
\end{equation*}
for $0 \leq n \leq N-1$, $(x_0,..,x_{n+1}) \in S^{n+2}$. We can now use $\vect{Y}$ to show the result:
\begin{equation*}
\begin{split}
    \langle e^{-\beta W(\vect{X})} \rangle &\stackrel{(i)}{=} \sum_{x_0,..,x_N \in S} \prob(X_0 = x_0) ({M_1})_{x_1 x_0} ({M_2})_{x_2 x_1}\cdots({M_N})_{x_N x_{N-1}} \times \\
&\qquad \qquad \  \times \frac{e^{-\beta E_1(x_0)}}{e^{-\beta E_0(x_0)}} \cdots \frac{e^{-\beta E_N(x_{N-1})}}{e^{-\beta E_{N-1}(x_{N-1})}} \\
&\stackrel{(ii)}{=} \sum_{x_0,\dots,x_N \in S} \prob(X_0=x_0) (\widehat{M}_N)_{x_0 x_1} \frac{e^{-\beta E_1(x_1)}}{e^{-\beta E_1(x_0)}}\cdots (\widehat{M}_1)_{x_{N-1} x_N} \times\\ 
&\qquad \qquad \ \times \frac{e^{-\beta E_N(x_{N})}}{e^{-\beta E_N(x_{N-1})}} \frac{e^{-\beta E_1(x_0)}}{e^{-\beta E_0(x_0)}}\cdots \frac{e^{-\beta E_N(x_{N-1})}}{e^{-\beta E_{N-1}(x_{N-1})}} \\
&\stackrel{(iii)}{=} \frac{Z_N}{Z_0} \sum_{x_0,\dots,x_N \in S}  \prob(Y_0 = x_N) (\widehat{M}_1)_{x_{N-1} x_N} \cdots (\widehat{M}_N)_{x_0 x_1} \\
&\stackrel{(iv)}{=} e^{-\beta \Delta F},
\end{split}
\end{equation*}
where we use the Markov property in $(i)$ and apply \eqref{def rev} in $(ii)$. In $(iii)$, we cancel the repeated terms coming from the definition of \smash{$(\widehat{M}_n)_{n=1}^N$} and from $e^{-\beta W(\vect{x})}$, since we have
\begin{eqnarray} \nonumber
\sum_{n=1}^N E_n(x_n)-E_n(x_{n-1}) 
& = & E_N(x_N) - E_1(x_0) + \sum_{n=1}^{N-1} E_n(x_n) - E_{n+1}(x_n) \\[4pt]
& = & E_N(x_N) - E_0(x_0) - W(\vect{x}) \label{1st law thermo}
\end{eqnarray}
which leads to  
\begin{equation*}
    \begin{split}
        &\prob(X_0 = x_0)  \frac{e^{-\beta E_1(x_1)}}{e^{-\beta E_1(x_0)}}..  \frac{e^{-\beta E_N(x_{N})}}{e^{-\beta E_N(x_{N-1})}} \frac{e^{-\beta E_1(x_0)}}{e^{-\beta E_0(x_0)}}.. \frac{e^{-\beta E_N(x_{N-1})}}{e^{-\beta E_{N-1}(x_{N-1})}} \\
        &= \frac{e^{- \beta E_0(x_0)}}{Z_0} e^{-\beta (\sum_{n=1}^ N E_n(x_n) - E_n(x_{n-1}))} e^{-\beta W(\vect{x})} \\
        &= \frac{1}{Z_0} e^{-\beta( E_0(x_0) + E_N(x_N) - E_0(x_0) - W(\vect{x}) + W(\vect{x}))} \\
        &= \frac{Z_N}{Z_0} \prob(Y_0 = x_N).
    \end{split}
\end{equation*}
Lastly, we apply the definition of $\Delta F$ and normalization of $\vect{Y}$ in $(iv)$.
\end{proof}

The main purpose of  proving Theorem \ref{Jarzynski's equality} is to show that Jarzynski's equality can be obtained under milder conditions than the ones that were considered before in decision-making. It should be noted that,
%main novelty of the proof of Theorem \ref{Jarzynski's equality} is that,
in contrast to the usual approaches like \cite{crooks1998nonequilibrium} where  $\vect{Y}$ is assumed to be the time reversal of $\vect{X}$, here it is just a convenient mathematical object for the proof. A similar approach to Jarzynksi's equality in the context of time-continuous Markov chains can be found in \cite{ge2007generalized}. Our approach to the discrete-time case in Theorem \ref{Jarzynski's equality} is much simpler, since we do not need measure-theoretic concepts nor smoothness assumptions.

%The main novelty of the proof of Theorem \ref{Jarzynski's equality} is that, in contrast to \cite{crooks1998nonequilibrium}, where  $\vect{Y}$ is assumed to be the time reversal of $\vect{X}$, here it is just a convenient mathematical object for the proof. A similar approach to Jarzynksi's equality in the context of time-continuous Markov chains can be found in \cite{ge2007generalized}. Our approach to the discrete-time case in Theorem \ref{Jarzynski's equality} is much simpler, since we do not need measure-theoretic concepts nor smoothness assumptions. 

% Also, they regard $\vect{Y}$ \cite[Defintion 2.3]{ge2007generalized} as insufficient and go on to define another Markov process (see Proposition 2.2 in \cite{ge2007generalized}) before deriving Jarzynski's equality. 

%Notice, by using Jensens' inequality on Jarzynski's equality \eqref{eq: jarz}, one obtains  
%\begin{equation*}
%\Delta F \leq \langle W(\vect{X}) \rangle,
%\end{equation*}
%which is a version of the second law \cite{jarzynski2011equalities}.

\section{Crooks' fluctuation theorem for Markov chains}
\label{sec:crooks}

The original derivation of Crooks' fluctuation theorem for Markovian dynamics \cite{crooks1998nonequilibrium,crooks1999entropy,crooks2000path} was carried out using a definition of work different from the one in \eqref{defi work}. In this section, we derive the theorem for Markov chains using \eqref{defi work} and comment on the difference between these approaches in the discussion.
%, has been accepted like that ever since (see for example \cite{jarzynski2011equalities}).
As discussed in the introduction, an additional hypothesis is needed for the result to hold in our setup. We derive Crooks' fluctuation theorem using this additional assumption in Theorem \ref{crooks fluctuation thm} and Corollary \ref{physics crooks} below, and then,  in Proposition \ref{counterexample}, we provide an example where this requirement is violated and the theorem is false everywhere.
%We address the influence of the missing constraint for its validity in the discussion. 

% Here, we prove Crooks' fluctuation theorem as a corollary to Theorem \ref{crooks fluctuation thm} below, a more general mathematical statement involving the process $\vect{Y}$ that we used in the proof of Theorem \ref{Jarzynski's equality}. Theorem \ref{crooks fluctuation thm} is itself a consequence of two propositions and their corollaries, that we present first. 
% We begin deriving, in Corollary \ref{physical core crooks}, a statement at the core of Crooks' fluctuation theorem whose proof is, aside from some details, in \cite{crooks2000path}. In order to do so, we include first, in Proposition \ref{core crooks}, a weaker statement.

Before proving the intermediate results and, finally, Crooks' fluctuation theorem, let us briefly sketch the procedure we follow throughout this section. We start with Proposition \ref{core crooks}, where we use the same Markov chain $\vect{Y}$ that we defined in the proof of Theorem \ref{Jarzynski's equality} and obtain, along similar lines, a more precise relation between $\vect{X}$ and $\vect{Y}$. In particular, between the probability of some realization of $\vect{X}$ and that of the same realization (with the events taking place in reversed order) of $\vect{Y}$. As a matter of fact, we show that, for a given $\vect{X}$, $\vect{Y}$ is (roughly) the only Markov chain fulfilling such a relation. Then, in Proposition \ref{Forward and reverse work equation}, we show how the driving signal of $\vect{X}$ and $\vect{Y}$ are related. In order to do so, we exploit the relation between the equilibrium distributions of $\vect{X}$ and $\vect{Y}$, which comes from the fact both Markov chains share the same equilibrium distributions. By a combining these two propositions, we reach a relation between the probability distributions of the driving signals of $\vect{X}$ and $\vect{Y}$ in Theorem \ref{crooks fluctuation thm}. Lastly, in Corollary \ref{physics crooks}, we impose an extra condition on the equilibrium distributions of $\vect{X}$ in order to obtain Crooks' fluctuation theorem in its usual form. We proceed now to show the details involved in the argument we just presented. We start by proving Proposition \ref{core crooks}.

\begin{prop}
\label{core crooks}
If $\vect{X}=(X_n)_{n=0}^N$ is a Markov chain on a finite state space $S$ whose initial distribution $p_0$ has non-zero entries and whose transition matrices $(M_n)_{n=1}^N$ are irreducible, then there exists a unique Markov chain $\vect{Y}=(Y_n)_{n=0}^N$ such that $Y_0 \sim p_N$, $(\widehat{M}_{n+1})_{xx} = (M_{N-n})_{xx}$ $\forall x \in S$, $0 \leq n \leq N-1$, where $(\widehat{M}_n)_{n=1}^N$ are the transition matrices of $\vect{Y}$, and for $\vect{x} = (x_0,x_1,\dots,x_N) \in S^{N+1}$,
\begin{equation}
\label{Crooks eq I}
\begin{split}
&\prob(X_1=x_1,\dots,X_N = x_N|X_0=x_0) \\
&=  \prob(Y_1=x_{n-1},\dots,Y_N=x_0|Y_0=x_N)\, e^{-\beta Q(\vect{x})}
\end{split}
\end{equation}
for any family of energies $\vect{E} =(E_n)_{n=0}^N$ of $\vect{X}$. Moreover, this unique $\vect{Y}$ satisfies
\begin{equation} 
\label{Crooks eq II}
\prob(\vect{X}=\vect{x}) \, = \, \prob(\vect{Y}=\vect{x}^R) \, e^{\beta( W(\vect{x}) - \Delta F)} ,
\end{equation}
where $W=W_{\vect{X},\vect{E}}$, $\Delta F= \Delta F_{\vect{X},\vect{E}}$, and $\vect{x}^R\coloneqq (x_N,\dots,x_0)$ denotes the reversal of $\vect{x}$. In particular, the probability of $\vect{X}$ following $\vect{x}$ is the same as the probability of $\vect{Y}$ following $\vect{x}^R$ if and only if $\Delta F = W(\vect{x})$.
\end{prop}

\begin{proof}
Consider the Markov chain $\vect{Y}=(Y_n)_{n=0}^N$ defined in the proof of Theorem \ref{Jarzynski's equality}, which is well-defined, since we have the same hypotheses. We can proceed in the same way as in Theorem \ref{Jarzynski's equality} to get
\begin{equation*}
\label{heateq}
    \begin{split}
        &\prob(X_1=x_1,\dots,X_N=x_N|X_0=x_0) \stackrel{(i)}{=} ({M_1})_{x_1 x_0}\cdots({M_N})_{x_N x_{N-1}}\\
        &\stackrel{(ii)}{=} \frac{e^{-\beta E_1(x_1)}}{e^{-\beta E_1(x_0)}}\cdots \frac{e^{-\beta E_N(x_N)}}{e^{-\beta E_N(x_{N-1})}}(\widehat{M}_N)_{x_0 x_1}\cdots(\widehat{M}_1)_{x_{N-1} x_N} \\ &\stackrel{(iii)}{=} e^{-\beta Q(\vect{x})} \prob(Y_1=x_{n-1},\dots,Y_N=x_0|Y_0=x_N),
    \end{split}
\end{equation*}
where we applied the Markov property of $\vect{X}$ in $(i)$, \eqref{def rev} in $(ii)$, and the definition of $Q$ plus the Markov property of $\vect{Y}$ in $(iii)$. This proves the existence of a Markov chain with the desired properties. Moreover, we have
\begin{equation*}
\begin{split}
    \prob(\vect{X} = \vect{x}) & \stackrel{(i)}{=} \frac{e^{-\beta E_0(x_0)}}{Z_0} e^{-\beta Q(\vect{x})} \prob(Y_1=x_{n-1},\dots,Y_N=x_0|Y_0=x_N) \\
    &\stackrel{(ii)}{=} \frac{Z_N}{Z_0}e^{-\beta (Q(\vect{x})-(E_N(x_N)-E_0(x_0)))} \, \prob(Y_0=x_N,\dots,Y_N=x_0)\\
    &\stackrel{(iii)}{=}e^{\beta (W(\vect{x})-\Delta F)} \, \prob(\vect{Y}=\vect{x}^R) ,
    \end{split}
\end{equation*}
where we applied \eqref{heateq}, the definition of conditional probability and the fact $X_0$ follows $p_0$ in $(i)$, the definition of conditional probability, the fact that $Y_0$ follows $p_N$ in $(ii)$, and both the definitions of $Q$ and $\Delta F$ plus \eqref{1st law thermo} in $(iii)$.

It remains to show the uniqueness of $\vect{Y}$. Assume $\vect{Z}=(Z_n)_{n=0}^N$ is a Markov chain with transition matrices \smash{$({M}'_n)_{n=1}^N$} such that $Z_0 \sim p_N$, \eqref{Crooks eq I} holds, and $({M}'_{n+1})_{xx} = (M_{N-n})_{xx}$ for all $x \in S$ and $0 \leq n \leq N-1$. Consider some $n$ such that $1<n<N$ and some $a,b \in S$ with $a \neq b$. By \eqref{Crooks eq I}, we have
\begin{equation*}
\begin{split}
       &\prob(X_1=a,\dots,X_{N-(n-2)}=a,X_{N-(n-1)}=b,\dots,X_N =b|X_0=a) e^{-\beta E_{N-(n-1)}(a)} \\ &=\prob(Z_1=b,\dots,Z_{n-1}=b,Z_n=a,\dots,Z_N=a|Z_0=b) e^{-\beta E_{N-(n-1)}(b)}.
       \end{split}
\end{equation*}
Applying the Markov property, the fact that $({M}'_{n+1})_{xx} = (M_{N-n})_{xx}$ $\forall x \in S$, $0 \leq n \leq N-1$, and the definition of $E_{N-(n-1)}$, we get
\begin{equation*}
    ({M}'_n)_{ab} = (M_{N-(n-1)})_{ba} \frac{p_{N-(n-1)}(a)}{p_{N-(n-1)}(b)} = (\widehat{M}_n)_{ab}.
\end{equation*}
Since the argument also works for $n=1$ and $n=N$, the case $a=b$ holds by definition, and $Y_0,Z_0 \sim p_N$, we have $\vect{Z}=\vect{Y}$.
\end{proof}

We say $\vect{X}$ is \emph{microscopically reversible} \cite{crooks2000path} if \eqref{Crooks eq I} is satisfied for $\vect{Y}$ being the  time reversal of $\vect{X}$, i.e. if the unique Markov chain $\vect{Y}$ that exists by Proposition \ref{core crooks} has initial distribution $p_N$ and transition matrices $(M_{N-n+1)})_{n=1}^N$. Notice, if this property holds, then \eqref{Crooks eq II} relates the probability of observing a realization $\vect{x}$ of a Markov chain $\vect{X}$ with that of observing the reversed realization $\vect{x}^R$ in the time reversal $\vect{Y}$ of $\vect{X}$, that is, when starting with the equilibrium distribution of the last environment and choosing according to the same conditional probabilities but in reversed order. This is the case if and only if the transition matrices of $\vect{X}$ satisfy detailed balance, as we show in the following lemma.

% We show in the following corollary that, under the same assumptions as before, $\vect{X}$ is microscopically reversible if and only if its transition matrices satisfy detailed balance. 

%Notice detailed balance holds for each transition matrix of $\vect{X}$ if and only if we have $\widehat{M}_{N-(n-1)} = M_{n}$ for $1 \leq n \leq N$ in Proposition \ref{core crooks}.
%Under this assumption, $\vect{Y}$ is the time reversal of $\vect{X}$ and \eqref{Crooks eq II} relates the probability of observing a realization of a Markov chain $\vect{X}$ with that of observing the same realization, in reversed order, in the time reversal of $\vect{X}$. We say $\vect{X}$ is \emph{microscopically reversible} \cite{crooks2000path} if \eqref{Crooks eq I} is fulfilled when taking $\vect{Y}$ to be the time reversal of $\vect{X}$.
%Notice, if $\vect{X}$ is a Markov chain whose transition matrices are irreducible, then $\vect{X}$ is microscopically reversible if and only if its transition matrices satisfy detailed balance.
%We state this stronger version of Proposition \ref{core crooks} in Corollary \ref{physical core crooks}.

\begin{lem}
\label{physical core crooks}
If $\vect{X}=(X_n)_{n=0}^N$ is a Markov chain on a finite state space $S$ with irreducible transition matrices $(M_n)_{n=1}^N$ and initial distribution with non-zero entries, then $M_n$ satisfies detailed balance for $1 \leq n \leq N$ if and only if $\vect{X}$ is microscopically reversible.
\begin{comment}
\begin{equation}
\label{micros rev}
    \prob(X_N = x_N,..,X_1=x_1|X_0=x_0) = \prob(Y_N=x_0,..,Y_1=x_{n-1}|Y_0=x_N) e^{-\beta Q(\vect{x})}
\end{equation}
for any family of energies of $\vect{X}$, $\vect{E} =(E_n)_{n=0}^N$, where $\vect{x} = (x_0,x_1,..x_N) \in S^{N+1}$ and $Q = Q_{\vect{X}}$.
\end{comment}
\end{lem}

\begin{proof}
If $\vect{X}$ satisfies detailed balance, that is, if each $p_n$ satisfies \eqref{det balance}, then by the definition of $\widehat{M}$ in \eqref{def rev}, we have \smash{$\widehat{M}_{n} = M_{N-(n-1)}$} for each $1\leq n\leq N$. Hence, in this case, the Markov chain $\vect{Y}$ constructed in Theorem \ref{Jarzynski's equality} and Proposition \ref{core crooks} is the time reversal of $\vect{X}$ and so $\vect{X}$ is microscopically reversible by definition.

It remains to show, if $\vect{X}$ is microscopically reversible, then its transition matrices satisfy detailed balance. Let $\vect{Y}$ be the time reversal of $\vect{X}$. Since, by assumption, $Y_0 \sim p_N$ and $\widehat{M}_{N-(n-1)} = M_{n}$, where  \smash{$(\widehat{M}_{n})_{n=1}^N$} are the transition matrices of $\vect{Y}$, we can follow the proof of Proposition \ref{core crooks} to get for $1 \leq n \leq N$
\begin{equation*}
       (M_n)_{ab} = (\widehat{M}_{N-(n-1)})_{ab} = (M_n)_{ba} \frac{p_n(a)}{p_n(b)}.
\end{equation*}
Thus, for each $1\leq n \leq N$, $p_n$ satisfies detailed balance with respect to $M_n$.
\end{proof}

In particular, this means that, if $\vect{X}$ satisfies detailed balance, then the time reversal $\vect{Y}$ of $\vect{X}$ satisfies \eqref{Crooks eq II}, that is
\begin{equation} 
\label{physics Crooks eq II}
\prob(\vect{X} = \vect{x})= \prob(\vect{Y}=\vect{x}^R) \, e^{\beta( W(\vect{x}) - \Delta F)} =  \prob(\vect{Y} = \vect{x}^R) \,  e^{\beta W^d(\vect{x})} 
\end{equation}
for any family of energies $\vect{E} =(E_n)_{n=0}^N$ of $\vect{X}$, where $W^d \coloneqq W^d_{\vect{X}}$ is the dissipated work of $\vect{X}$.
% In particular, $\prob(\vect{X}=\vect{x}) = \prob(\vect{Y}=\vect{x}^R)$ if and only if $\Delta F = W(\vect{x})$.
Thus, in this case, dissipated work is an unambiguous measure of the discrepancy between the probability of observing a realization of $\vect{X}$ and the probability of observing the same trajectory in reversed order in the time reversal of $\vect{X}$.  We have, hence, an unambiguous measure of \emph{hysteresis} (see Section \ref{sec:decision}).

% We will denote by $\vect{x}^R$ the vector consisting of the components in $\vect{x}$ in inverted order, that is, for any $\vect{x} = (x_0,x_1,..,x_N) \in S^{N+1}$ we have $\vect{x}^R \coloneqq (x_N,..,x_1,x_0) \in S^{N+1}$.

Before showing Theorem \ref{crooks fluctuation thm} and Crooks' fluctuation theorem, we relate the work of $\vect{X}$ with that of its time reversal.
%relates the work along $\vect{x}$ for $\vect{X}$ to the work along $\vect{x}^R$ for $\vect{Y}$.

%Given any $\vect{x} = (x_0,x_1,..,x_N) \in S^{N+1}$, we denote by $\vect{x}^R$ the vector consisting of the same components in $\vect{x}$ but in inverted order, that is, $\vect{x}^R \coloneqq (x_N,..,x_1,x_0) \in S^{N+1}$. The third partial result, Proposition \ref{Forward and reverse work equation}, relates the work along any $\vect{x}$ for the forward process to the work along $\vect{x}^R$ for the backward process.

\begin{prop}
\label{Forward and reverse work equation}
If $\vect{X} = (X_n)_{n=0}^N$ is a Markov chain on a finite state space $S$ with initial distribution with non-zero entries and irreducible transition matrices, then there exists a Markov chain $\vect{Y}=(Y_n)_{n=0}^N$ and a constant $k\in\mathbb R$ such that
\begin{equation}
\label{forward work = back work + sth}
    W_{\vect{Y},\vect{\widehat{E}}}(\vect{x}^R) = -W_{\vect{X},\vect{E}}(\vect{x}) + E_1(x_0)-E_0(x_0) + k \qquad \forall \vect{x} \in S^{N+1}, 
\end{equation}
where $\vect{E}$ and $\vect{\widehat{E}}$ are families of energies of $\vect{X}$ and $\vect{Y}$, respectively. 
% More precisely, the constant $k$ is given by $(\widehat{E}_N - E_1) - (\widehat{E}_0-E_N)$. 
Moreover, if the stationary distribution $p_1$ of $M_1$ coincides with the initial distribution $p_0$ of $\vect{X}$, then there exists a constant $k\in\mathbb R$ such that
\begin{equation}
\label{opposite work}
    W_{\vect{Y},\vect{\widehat{E}}}(\vect{x}^R) = -W_{\vect{X},\vect{E}}(\vect{x}) + k \qquad \forall \vect{x} \in S^{N+1}.
\end{equation}
% where $k' \coloneqq (\widehat{E}_N-E_0) - (\widehat{E}_0-E_N) \in \mathbb{R}$.
\end{prop}

\begin{proof}
Let $\vect{Y}=(Y_n)_{n=0}^N$ be the Markov chain defined in the proof of Theorem \ref{Jarzynski's equality}. Since work is well-defined (up to a constant) for both $\vect{X}$ and $\vect{Y}$, by Lemma \ref{energy rela} (see Appendix \ref{appendix}) we can use the relation between the energy functions of both chains in \eqref{const diff} to show \eqref{forward work = back work + sth}. We have
\begin{equation*}
\begin{split}
    W_{\vect{Y},\vect{\widehat{E}}}(\vect{x}^R) &= \sum_{n=0}^{N-1} \widehat{E}_{n+1}(x^R_n) - \widehat{E}_n(x^R_n) \\
    &\stackrel{(i)}{=} \sum_{n=1}^{N-1} \widehat{E}_{n+1}(x^R_n) - \widehat{E}_n(x^R_n) + \widehat{c}\\
    &\stackrel{(ii)}{=} \sum_{n=1}^{N-1} E_{N-n}(x_{N-n}) - E_{N-(n-1)}(x_{N-n}) + k\\
    &\stackrel{(iii)}{=} \sum_{m=1}^{N-1} E_m(x_m) - E_{m+1}(x_m) + k\\
    & = -W_{\vect{X},\vect{E}}(\vect{x}) + E_1(x_0) - E_0(x_0) + k,
    \end{split}
\end{equation*}
where in $(i)$ we defined $\widehat{c} \coloneqq \widehat{E_1} - \widehat{E_0}$, which is a constant since $\widehat{E}_{0} = E_{N} + k_{0}$ and $\widehat{E}_{1} = E_{N} + k_{1}$ by Lemma \ref{energy rela} (see Appendix \ref{appendix}). 
In $(ii)$, we applied the definition of $\vect{x}^R$ and \eqref{const diff}, cancelled the repeated $k_n$, defined as in Lemma \ref{energy rela}, for all $1<n<N$ and introduced $k \coloneqq k_N - k_1 +\widehat{c}= (\widehat{E}_N - E_1) - (\widehat{E}_1-E_N) +(\widehat{E}_1 - \widehat{E}_0) =(\widehat{E}_N - E_1) - (\widehat{E}_0-E_N)$. In $(iii)$, we rewrite the sum in terms of $m \coloneqq N-n$.

For the second statement, notice, if $p_0$ is the stationary distribution of $M_1$, then there exists a constant $c$ such that $E_1=E_0 + c$ by \eqref{ener shift}. Thus, we have $k \coloneqq (\widehat{E}_N-E_0) - (\widehat{E}_0-E_N) = (E_1 - E_0) + (\widehat{E}_N-E_1) - (\widehat{E}_0-E_N) =c +k'=E_1(x_0)-E_0(x_0) + k'$ for all $x_0 \in S$, where $k'$ is the constant in \eqref{forward work = back work + sth}.
\end{proof}

If detailed balance holds, then \eqref{forward work = back work + sth} and \eqref{opposite work} relate the work along a realization $\vect{x}$ of $\vect{X}$ with the reversed realization $\vect{x}^R$ of its time reversal $\vect{Y}$. More precisely, we obtain the following corollary.

\begin{coro}[When work is odd under time reversal]
\label{work coro}
If all transition matrices of $\vect{X}$ in Proposition \ref{Forward and reverse work equation} satisfy detailed balance and $\vect{Y}$ is the time reversal of $\vect{X}$, then the constants $k$ in \eqref{forward work = back work + sth} and \eqref{opposite work} can be taken to be zero.
% If $\vect{X}=(X_n)_{n=0}^N$ is a Markov chain on a finite state space $S$ whose initial distribution $p_0$ and transition matrices $(M_n)_{n=1}^N$ have non-zero entries and whose transition matrices satisfy detailed balance, $\vect{E}\coloneqq (E_n)_{n=0}^N$ is a family of energies of $\vect{X}$ and $\vect{Y}=(Y_n)_{n=0}^N$ is the time reversal of $\vect{X}$, then there exists a family of energies of $\vect{Y}$, $\vect{\widehat{E}}\coloneqq (\widehat{E}_n)_{n=0}^N$, such that
% \begin{equation}
% \label{physics forward work = back work + sth}
%     W_{\vect{Y},\vect{\widehat{E}}}(\vect{x}^R) = -W_{\vect{X},\vect{E}}(\vect{x}) + E_1(x_0)-E_0(x_0)
% \end{equation}
% $\forall \vect{x} \in S^{N+1}$. Moreover, if $p_0$ is the stationary distribution of $M_1$, then we can choose $\vect{\widehat{E}}\coloneqq (\widehat{E}_n)_{n=0}^N$ such that work is odd under time reversal, that is, 
% \begin{equation}
% \label{physics opposite work}
%     W_{\vect{Y},\vect{\widehat{E}}}(\vect{x}^R) = -W_{\vect{X},\vect{E}}(\vect{x})
% \end{equation}
% $\forall \vect{x} \in S^{N+1}$.
\end{coro}

\begin{proof}
For the constant in \eqref{forward work = back work + sth} we simply choose $\widehat{E}_N=E_1$ and $\widehat{E}_0=E_N$, and for the constant in \eqref{opposite work} we choose \smash{$\widehat{E}_N=E_0$} and $\widehat{E}_0=E_N$, which we can do since $p_0$ is the stationary distribution of $M_1$ and, by  Lemma \ref{energy rela} (see Appendix \ref{appendix}) also that of \smash{$\widehat{M}_N$}.
\end{proof}

%Notice, although we keep the constants in the statement of Proposition \ref{Forward and reverse work equation} for generality, $k=(\widehat{E}_N-E_1) - (\widehat{E}_0-E_N)=0$ in thermodynamics as we have $\widehat{E}_n=E_{N-(n-1)}$ for $1 \leq n \leq N$ and $\widehat{E}_0=\widehat{E}_1$. Moreover, in case $p_0$ is a stationary distribution of $M_1$, we have $E_1=E_0$ in thermodynamics and, thus, $k'=0$.

\begin{rem}
\label{const in thermo}
Note that choosing the energy functions in Corollary \ref{work coro} is unnecessary whenever both $\vect{X}$ and $\vect{Y}$ are thermodynamic processes. Although energy is defined only up to a constant in thermodynamics, it would make no sense to pick the constants differently when dealing with a system where the same dynamics occur more than once. Thus, there, we have $\widehat{E}_n=E_{N-(n-1)}$ for $1 \leq n \leq N$, $\widehat{E}_0=\widehat{E}_1$ and, in case $p_0$ is a stationary distribution of $M_1$, \smash{$E_1=E_0=\widehat{E}_N$}. In particular, when taking $\vect{Y}$ to be the time reversal of $\vect{X}$ in thermodynamics, we always have $k=0$ in \eqref{forward work = back work + sth} and, in case $p_0$ is a stationary distribution of $M_1$, $k=0$ in \eqref{opposite work}.
\end{rem}

The non-constant term $E_1(x_0)-E_0(x_0)$ in \eqref{forward work = back work + sth}, which remains even when $\vect{X}$ satisfies detailed balance, follows from an asymmetry between $\vect{X}$ and $\vect{Y}$.
%, which seems to be absent in the literature aside from an unpublished note \cite{gujrati2019consequences}
In particular, $\vect{X}$ goes from $E_0$ to $E_N$, whereas $\vect{Y}$ goes from $E_N$ to $E_1$, because while $X_0 \sim p_0 \propto e^{-\beta E_0}$ and the stationary distribution of $M_N$ is $p_N \propto e^{-\beta E_N}$, which is also the initial distribution of $\vect{Y}$, the stationary distribution of the final transition matrix  $\widehat{M}_N$ of $\vect{Y}$ is $p_1 \propto e^{-\beta E_1}$ (and not $p_0$). Also, while $\vect{X}$ may begin with a change in the energy function, since $p_1 \neq p_0$ is allowed, $\vect{Y}$ does not, as $p_N$ is both the initial distribution of $\vect{Y}$ and the stationary distribution of $\widehat{M}_1$. An example can be found in Figure \ref{asymmetry work}. This asymmetry is erased if we assume that $p_0$ is the stationary distribution of $M_1$, in which case, for Markov chains $\vect{X}$ that satisfy detailed balance, the work along any realization of $\vect{X}$ has the opposite sign of the work along the reversed realization of $\vect{Y}$. That is, thermodynamic work becomes \emph{odd under time reversal}.
%Notice, although it is usually stated \cite{crooks1998nonequilibrium,crooks2000path,jarzynskia2008nonequilibrium}, this does \emph{not} hold without the additional assumption $p_1=p_0$ (c.f. Figure \ref{asymmetry work} and the discussion).

The following theorem contains Crooks' fluctuation theorem as the special case when $\vect{X}$ satisfies detailed balance (see Corollary \ref{physics crooks} below).

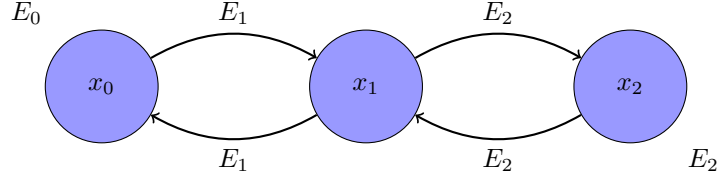
\begin{figure}[!tb]
    \centering
\begin{tikzpicture}[auto,node distance=2.8cm]
    \node[main node] (1) {$x_0$};
    \node[main node] (2) [right = 2cm of 1]  {$x_1$};
     \node[main node] (3) [right = 2cm of 2] {$x_2$};
     
    \node at (-1,0.98) {$E_0$};
     \node at (8,-0.98) {$E_2$};

    \path[draw,thick,->]
    (1) edge [bend left]  node {$E_1$} (2)
    (2) edge [bend left] node {$E_2$} (3)
    (2) edge [bend left]  node {$E_1$} (1)
    (3) edge [bend left]  node {$E_2$} (2);
    %(3) edge node {} (1);
\end{tikzpicture}
\caption{Simple example showing how the asymmetry between $\vect{X}$ and its time reversal $\vect{Y}$ manifests in thermodynamics. We consider $\vect{x}=(x_0,x_1,x_2)$ a trajectory, $(E_0,E_1,E_2)$ with $E_1(x_0) \not = E_0(x_0)$ the energy functions of $\vect{X}$ (upper line of arrows) and, by Remark \ref{const in thermo}, $(E_2,E_2,E_1)$ the energy functions of $\vect{Y}$ (bottom line of arrows). We have $W_{\vect{X}}(\vect{x})= E_2(x_1)-E_1(x_1) + E_1(x_0) - E_0(x_0)$ and $W_{\vect{Y}}(\vect{x}^R)= E_1(x_1)-E_2(x_1)$. As a result $W_{\vect{Y}}(\vect{x}^R)= -W_{\vect{X}}(\vect{x}) + E_1(x_0) - E_0(x_0)$, in accordance with \eqref{forward work = back work + sth} and Corollary \ref{work coro}. Thus, thermodynamic work is \emph{not} odd under time reversal in general.}
\label{asymmetry work}
\end{figure}

\begin{comment}
\begin{prop}
The following holds
\begin{equation}
\frac{\prob^F(w| X_0=x_0)}{\prob^B(-w + E_1(x_0) - E_0(x_0)|Y_N=x_0)} = e^{\beta(w-\Delta F)}
\end{equation}
where $\prob^F(w)|X_0=x_0):= \prob(W(\vect{x})=w| (X_n)_{0 \leq n \leq N}, X_0=x_0 )$ and $\prob^B(-w+ E_1(x_0) - E_0(x_0)|Y_N=x_0):= \prob(W(\vect{x}^R)=-w+ E_1(x_0) - E_0(x_0)| (Y_n)_{0 \leq n \leq N}, Y_N=x_0 )$ where given $\vect{x} = (x_0,x_1,..,x_N)$ we define $\vect{x}^R=(x_N,..,x_1,x_0)$. In particular, if $E_1=E_0$ then 
\begin{equation}
\label{work theo}
    \frac{\prob^F(w)}{\prob^B(-w)} = e^{\beta (w-\Delta F)}
\end{equation}
\end{prop}

\begin{proof}
\begin{equation}
\begin{split}
    \prob^F(w|X_0=x_0) &= \sum_{ \vect{x} \in x_0 \times S^N} \prob(\vect{x}| W(\vect{x})=w) \\
    &=^{(i)} e^{\beta(w-\Delta F)} \sum_{\vect{x} \in x_0 \times S^N} \prob(\vect{x}^R| W(\vect{x})=w) \\
    &=^{(ii)} e^{\beta(w-\Delta F)} \sum_{\vect{x} \in x_0 \times S^N} \prob(\vect{x}^R| W(\vect{x}^R)= -w +E_1(x_0)-E_0(x_0)) \\
    &= e^{\beta(w-\Delta F)} \prob^B(-w+E_1(x_0)-E_0(x_0)|Y_N=x_0)
    \end{split}
\end{equation}
In $(i)$ we applied \eqref{Crooks eq II}. In $(ii)$ we applied Proposition \ref{Forward and reverse work equation}. 
\end{proof}
\end{comment}

\begin{theo}
\label{crooks fluctuation thm}
If $\vect{X} = (X_n)_{n=0}^N$ is a Markov chain on a finite state space $S$
whose initial distribution $p_0$ has non-zero entries, whose transition matrices $(M_n)_{n=1}^N$ are irreducible,
%whose initial distribution $p_0$ and transition matrices $(M_n)_{n=1}^N$ have non-zero entries,
and where $p_0$ is the stationary distribution of $M_1$, that is $p_1=p_0$, then there exists a Markov chain $\vect{Y}=(Y_n)_{n=0}^N$ and a constant $k\in \mathbb R$ such that
\begin{equation}
\label{eq: crooks}
    % \frac{\prob^F(W=w)}{\prob^B(W=-w+k')} = e^{\beta(w- \Delta F)}
    \frac{\prob(W_{\vect{X},\vect{E}}=w)}{ \prob(W_{\vect{Y},\vect{\widehat{E}}}=-w+k)} = e^{\beta(w- \Delta F)} \qquad \forall w \in \supp (\prob(W_{\vect{X},\vect{E}})),%W_{\vect{X},\vect{E}}(S^{N+1}),
\end{equation}
%\textcolor{red}{for all $w \in W_{\vect{X},\vect{E}}(S^{N+1})$ such that $\prob(W_{\vect{X},\vect{E}}=w)>0$,}
where 
% $\prob^F(W=w) \coloneqq \prob(W_{\vect{X},\vect{E}}=w)$, $\prob^B(W=w) \coloneqq \prob(W_{\vect{Y},\vect{\widehat{E}}} $ $=w)$, 
$\vect{E} = (E_n)_{n=0}^N$ and $\vect{\widehat{E}} = (\widehat{E}_n)_{n=0}^N$ are families of energies of $\vect{X}$ and $\vect{Y}$, respectively, and $\supp (\prob(W_{\vect{X},\vect{E}}))$ denotes the support of the probability distribution of $W_{\vect{X},\vect{E}}$, that is, the values that can be taken by $W_{\vect{X},\vect{E}}$ with non-zero probability.
\end{theo}

\begin{proof}
Let $\vect{Y}=(Y_n)_{n=0}^N$ be the Markov chain defined in the proof of Theorem \ref{Jarzynski's equality}. Note that for any family of energies  $\vect{E}$ of $\vect{X}$ there exists a constant $c$ such that $E_1=E_0+c$ by \eqref{ener shift}, since $p_0$ is the stationary distribution of $M_1$. Given some $ w \in W_{\vect{X},\vect{E}}(S^{N+1})$, we have
\begin{equation*}
%\label{Pf and PB}
\begin{split}
    \prob(W_{\vect{X},\vect{E}}=w) &=
    \sum_{\vect{x} \in W_{\vect{X},\vect{E}}^{-1}(w)} \prob((X_0,..,X_N)=\vect{x}) \\
    &\stackrel{(i)}{=} e^{\beta(w-\Delta F)} \sum_{\vect{x} \in W_{\vect{X},\vect{E}}^{-1}(w)} \prob((Y_0,..,Y_N)=\vect{x}^R) \\
    &\stackrel{(ii)}{=} e^{\beta(w-\Delta F)} \sum_{\vect{x} \in S^{N+1}:\text{ } W_{\vect{Y},\vect{\widehat{E}}}(\vect{x}^R)=-w+k} \prob((Y_0,..,Y_N)=\vect{x}^R) \\
    &= e^{\beta(w-\Delta F)} \ \prob(W_{\vect{Y},\vect{\widehat{E}}}=-w+k),
    \end{split}
\end{equation*}
 where we applied \eqref{Crooks eq II} in $(i)$, and Proposition \ref{Forward and reverse work equation} plus the fact that $E_1=E_0+c$ in $(ii)$. To get \eqref{eq: crooks}, it remains to show that \smash{$\prob(W_{\vect{Y},\vect{\widehat{E}}}=-w+k)>0$} for all $w \in \supp (\prob(W_{\vect{X},\vect{E}}))$. By definition, there exists some $\vect{x} \in S^{N+1}$ such that $W_{\vect{X},\vect{E}}(\vect{x})=w$ and $\prob(\vect{X}=\vect{x})>0$. By Proposition \ref{Forward and reverse work equation}, $W_{\vect{Y},\vect{\widehat{E}}}(\vect{x}^R) = -w + k$. Since $\prob(\vect{X}=\vect{x})>0$
 %the entries of the transition matrices from $\vect{X}$ are non-zero by hypothesis
 and the entries of the (unique) stationary distributions of $\vect{X}$ are also non-zero by Lemma \ref{irred prop}, we can use \eqref{def rev} plus the Markov property for both $\vect{X}$ and $\vect{Y}$ to show $\prob(\vect{Y}=\vect{x}^R)>0$, implying $\prob(W_{\vect{Y},\vect{\widehat{E}}}=-w+k)>0$.
\end{proof}

% If we assume each transition matrix of $\vect{X}$ fulfills detailed balance, then we have $\widehat{M}_{N-(n-1)} = M_{n}$ for $1 \leq n \leq N$. Under this assumption, $\vect{Y}$ is the time reversal of $\vect{X}$ and, since $k'=0$ in thermodynamics, \eqref{eq: crooks} relates the probability of observing certain work value in a realization of a Markov chain with that of observing the opposite work value in a realization of its time reversal. We state this stronger version of Theorem \ref{crooks fluctuation thm} in Corollary \ref{physics crooks}.

As the special case when each transition matrix of $\vect{X}$ in Theorem \ref{crooks fluctuation thm} satisfies detailed balance, we obtain Crooks' fluctuation theorem modified by the additional assumption of $p_0=p_1$.

% If detailed balance holds for each transition matrix of $\vect{X}$, then we can take the time reversal of $\vect{X}$ as $\vect{Y}$ in Theorem \ref{crooks fluctuation thm} and \eqref{eq: crooks} relates the probability of observing a certain work value in a realization of $\vect{X}$ with that of observing the work value with the opposite sign in a realization of its time reversal $\vect{Y}$:

\begin{coro}[Crooks' fluctuation theorem for Markov chains]
\label{physics crooks}
If all transition matrices $(M_n)_{n=1}^N$ of the Markov chain $\vect{X}=(X_n)_{n=0}^N$ in Theorem \ref{crooks fluctuation thm} satisfy detailed balance, then \eqref{eq: crooks} holds with $k=0$ and $\vect{Y}$ being the time reversal of $\vect X$, that is the Markov chain with initial distribution $p_N$ and transition matrices $(M_{N-(n-1)})_{n=1}^N$. 
% $p_0$ is the stationary distribution of $M_1$, $\vect{E}\coloneqq (E_n)_{n=0}^N$ is a family of energies of $\vect{X}$ and $\vect{Y}=(Y_n)_{n=0}^N$ is the time reversal of $\vect{X}$,
% %is the Markov chain whose initial distribution is $p_N$, the stationary distribution of $M_N$, and whose transition matrices are $(M_{N-(n-1)})_{n=1}^N$,
% then there exists a family of energies of $\vect{Y}$, $\vect{\widehat{E}}\coloneqq (\widehat{E}_n)_{n=0}^N$, such that we have
% \begin{equation}
% \label{eq: crooks physics}
%     \frac{\prob^F(W=w)}{\prob^B(W=-w)} = e^{\beta (w-\Delta F)}
% \end{equation}
% $\forall w \in W_{\vect{X},\vect{E}}(S^{N+1})$, where $\prob^F(W=w) \coloneqq \prob(W_{\vect{X},\vect{E}}=w)$ and $\prob^B(W=w) \coloneqq \prob(W_{\vect{Y},\vect{\widehat{E}}} =w)$.
\end{coro}

\begin{proof}
As can be seen from the proof of Theorem \ref{crooks fluctuation thm} and Corollary \ref{work coro}, in case of detailed balance, we can choose $\vect{Y}$ to be the time reversal of $\vect{X}$. Moreover, the origin of the constant $k$ in Theorem \ref{crooks fluctuation thm} is Equation \eqref{opposite work}. By Corollary \ref{work coro}, this constant can be set to zero if the transition matrices of $\vect{X}$ satisfy detailed balance.
\end{proof}

Notice, in most of the literature on Crooks' fluctuation theorem one writes $\mathbb P^F(W)$ for the probability $\mathbb P(W_{\vect{X},\vect{E}})$ of the work along the so-called \emph{forward process} $\vect{X}$ and $\mathbb P^B(W)$ for the probability $\mathbb P(W_{\vect{Y},\vect{\hat E}})$ of the work along the so-called \emph{backward process} $\vect{Y}$ (the time reversal of $\vect{X}$), so that, by Corollary \ref{physics crooks}, under detailed balance, Equation \eqref{eq: crooks} reads
\begin{equation} \label{eq: crooks physics}
\frac{\prob^F(W=w)}{ \prob^B(W=-w)} = e^{\beta(w-\Delta F)} \, .
\end{equation}

% Notice, choosing the energy functions in Corollary \ref{physics crooks} is unnecessary in thermodynamics, as we directly have  $k=0$ in \eqref{eq: crooks} by Remark \ref{const in thermo}.

The condition that $p_0$ is the stationary distribution of $M_1$ is not necessary in Crooks' original work \cite{crooks1998nonequilibrium,crooks1999entropy,crooks1999excursions,crooks2000path}.
%neither in further research on the topic (see, for example, \cite{jarzynskia2008nonequilibrium} or \cite{jarzynski2011equalities}).
Nonetheless, it is fundamental in our approach:
%a relevant remark:
Crooks' fluctuation theorem can even be false for every work value if $p_0$ is not the stationary distribution of $M_1$, as we show in Proposition \ref{counterexample}.

\begin{prop}
\label{counterexample}
If $p_0$ is not the stationary distribution of $M_1$, then there exist Markov chains where Crooks' fluctuation theorem is false everywhere, despite the other assumptions in Theorem \ref{crooks fluctuation thm} and Corollary \ref{physics crooks} being fulfilled.
\end{prop}
\begin{proof}
We consider a state space with three components $S=\{A,B,C\}$, a Markov chain with two steps $\vect{X}=(X_0,X_1)$ and $a,b,c >0$ where $b<c$ and $a+b+c=1$. We take $E_0$ as energy function associated to $X_0$ where $E_0(A)=\log \frac{1}{a}$, $E_0(B)=\log \frac{1}{b}$ and $E_0(C)=\log \frac{1}{c}$ and $E_1$ as energy function associated to $X_1$ where $E_1(A)=E_0(A)$, $E_1(B)=E_0(C)$ and $E_1(C)=E_0(B)$. Taking $\beta =1$, we get both free energies $F_1$ and $F_0$ equal to one. Notice we have $p_0=(a,b,c)$, with non-zero entries, and $p_1=(a,c,b)$. Notice, also, $W_{\vect{X}}(S^2) = \{w_A,w_B,w_C\}$, where $w_C \coloneqq E_1(C)-E_0(C)>0$, $w_B \coloneqq E_1(B)-E_0(B)<0$ and $w_A \coloneqq E_1(A)-E_0(A)=0$. We fix $\widehat{E}_1=\widehat{E}_0=E_1$. We can easily see that \eqref{eq: crooks physics} is not defined for $w=w_B,w_C$, and that it is false, although defined, for $w=w_A$. We have $W_{\vect{Y}}(\vect{x})=0$ $\forall \vect{x} \in S^2$, since $p_1$ is both the starting distribution of time reversal of $\vect{X}$ and the stationary distribution of its only transition matrix. Thus, we have $\prob^F(W=w_C)=\prob(X_0=C)=c>0$ and $\prob^B(W=-w_C)=0$ which means \eqref{eq: crooks physics} is not defined at $W=w_C$. We get, analogously, it is not defined for $W=w_B$. For $W=w_A$, we have $\prob^B(W=-w_A)=1$ and 
\[
\prob^F(W=w_A)=\prob(X_0=A)= a < 1= e^{\beta w_A}
\]
which means \eqref{eq: crooks physics} is defined and false there. Although the argument is independent of the transition matrix for $\vect{X}$, fix
\begin{equation*}
M_1 \coloneqq 
    \begin{pmatrix}
a & a & a\\
c & c & c \\
b & b & b
\end{pmatrix}
\end{equation*}
for completeness, since it has non-zero entries and fulfills detailed balance with respect to $p_1$.
\end{proof}

\section{Discussion: Application to decision-making}
\label{sec:decision}

%In fact, t
The bridge connecting thermodynamics and decision-making is an optimization principle directly inspired by the \emph{maximum entropy principle} \cite{jaynes1957information,jaynes2003probability,wolpert2006information,genewein2015bounded}. In particular, given a finite set of possible choices $S$, the optimal behaviour is given by the distribution $p \in \mathbb P_S$ that optimally trades utility and uncertainty according to the following optimization principle:
\begin{equation}
	\label{bounded rat prin II}
	p = \arg\max_{ q \in \mathbb P_S} \Big\{H(q)\Big| \mathbb E_q[U] \geq U_0\Big\},
\end{equation}
%$\mathbb E[~\cdot~]$ denotes the expectation operator
where $\mathbb P_S$ is the set of probability distributions over $S$, $H$ is the Shannon entropy, $\mathbb E_q[~\cdot~]$ denotes the expected value over $q \in \mathbb P_S$ and $U:S \to \mathbb R$ is a \emph{utility} function, that is, a function that assigns larger values to options in $S$ that are preferred by the decision-maker. Note that the main difference between \eqref{bounded rat prin II} and the maximum entropy principle is the the substitution of an energy function $E: S \to \mathbb R$ by a utility $U$, which behaves as a negative energy function $U=-E$ (in the sense that it is the \emph{force} that opposes uncertainty). As a result of their similarity, both principles yield the same result, namely, the Boltzmann distribution 
\begin{equation*}
	%\label{boltzmann}
	p(x) = \frac{1}{Z} e^{\beta U(x)}=\frac{1}{Z} e^{-\beta E(x)} \text{ } \forall x \in S,
\end{equation*}
where $\beta$ is a trade-off parameter between uncertainty and utility/energy, and $Z$ is a normalization constant.

The analogy between thermodynamics and decision-making can be taken further by considering not only their optimal distribution but how they transition between different distributions in the path towards the optimal one. Here, the notion of uncertainty is useful again, although this time it is relative to the optimal distribution $p$. More specifically, we can think of the transitions the decision-maker undergoes as being driven by the reduction of uncertainty with respect to the optimal distribution $p$, which can be modelled by the dual of $p$-majorization \cite{joe1990majorization}. In this approach, given $q,q' \in \mathbb P_S$, the decision-maker transitions from $q$ to $q'$, which we denote by $q \preceq_p q'$, if $q'$ is \emph{closer} to $p$ than $q$ (see \cite{joe1990majorization} for a rigorous definition of the dual of $\preceq_p$ and, hence, of \emph{closer}). For us, the important fact about $\preceq_p$ is that, as it turns out \cite[Theorem 2]{gottwald2019bounded}, the transitions that are allowed by $\preceq_p$ are precisely the ones that result from applying a transition matrix that has $p$ as stationary distribution. More precisely,
\begin{equation}
	\label{evo to p}
	q \preceq_p q' \iff q'= M_p q,
\end{equation}
where $M_p$ is a matrix whose rows are normalized and fulfills $M_p p = p$, that is, a \emph{stochastic} matrix for which $p \in \mathbb P_S$ is a stationary distribution \cite{levin2017markov}. Importantly, the transition matrix assumption is common in the study of both thermodynamic and decision-making systems \cite{crooks1998nonequilibrium,crooks2000path}.

In our current study, the situation we have in mind is that of a decision-maker that has to take a sequence of decisions under varying environmental conditions. We model this as a stochastic process which behaves like a Markov chain and introduce, starting from a Markov chain, the thermodynamic tools we use to describe it: energy, partition function, free energy, work, heat and dissipated work. The behavior of the decision-maker then corresponds to a decision vector $(x_0,x_1,\dots,x_n)$ collected over $n$ potentially different environments. 
In this most general decision-making scenario, where the environment is changing over time, the optimal distribution $p$ is changing as well. Thus, we can regard the  decision-making process as a sequence of transition matrices $M_p$, where each $M_p$ corresponds to a particular environment with optimal response $p$. We could imagine, for example, a gradient descent learner that would converge to $p$ for any given environment, presuming we allow for sufficient gradient update steps. Otherwise the gradient learner, or any other optimization-based decision-making agent (eg. following a Metropolis-Hastings optimization scheme), would lag behind the environmental changes and the environment would outpace the learner. In this general decision-making scenario we can then study the relation between the optimal behavior and the non-optimal one by fluctuation theorems \cite{seifert2012stochastic,jarzynski2011equalities,jarzynskia2008nonequilibrium} like Jarzynski's equality \cite{jarzynski1997equilibrium,jarzynski2004nonequilibrium,jarzynski1997nonequilibrium} and Crooks' fluctuation theorem \cite{crooks1999entropy,crooks1998nonequilibrium}. 
While we focus on these two fluctuation theorems in our current study, similar arguments may be suitable to transfer other fluctuation theorems that have been considered in the thermodynamic literature (see for example \cite{seifert2012stochastic}) to a decision-making scenario.

\paragraph{Jarzynski's equality in decision-making.} Although the strong requirements in Lemma \ref{physical core crooks} have been used in \cite{crooks1998nonequilibrium} to derive Jarzynski's equality through \eqref{physics Crooks eq II} and were assumed in the only approach we know to it in decision-making \cite{grau2018non}, weaker assumptions which do not involve the time reversal of $\vect{X}$ are sufficient (cf. Theorem \ref{Jarzynski's equality}). The same properties have been used to derive Jarzynski's equality following a different method in \cite{jarzynski1997equilibrium}.

	\paragraph{Example application: Jarzynski's theorem.}
	\label{appli jarz}
	In a decision scenario, the energy becomes a loss function that a decision-maker is trying to minimize. If this loss function changes over time, we can conceptually distinguish changes in loss that are induced externally by changes in the environment (e.g. given data), from changes in loss due to internal adaptation when a learning system changes its parameter settings. The externally induced changes in loss correspond to the physical concept of work and drive the adaptation process. Hence, we can consider the decision-theoretic equivalent of physical work as a driving signal: the (negative) surprise experienced by the decision-maker, given that it adds the (negative) surprise it experiences at each step (which can be quantified by the difference in energy/utility evaluated at the decision-maker's state when the environment changes) \cite{grau2018non,hack2022thermodynamic}. With this in mind, we can use Jarzynski's equality to obtain a bound on the  decision-maker's expected surprise. In particular, by applying Jensens' inequality on Jarzynski's equality, one obtains  
	\begin{equation}
		\label{second law}
		\Delta F \leq \mathbb E\big[ W(\vect{X}) \big].
	\end{equation}
	(Note this is a version of the second law of thermodynamics \cite{jarzynski2011equalities}.)
	Hence, \eqref{second law} provides a bound on the expected surprise. While a similar bound has been previously pointed out for decision-making systems \cite{grau2018non}, here we re-derive it under a novel energy protocol and with weaker assumptions regarding time-reversibility.

\paragraph{Crooks' theorem in decision-making.} Even though the assumption that $p_0$ must be the stationary distribution of $M_1$ seems to restrict the applicability of Crooks' fluctuation theorem in our decision-theoretic setup when compared to the usual thermodynamics one (see Corollary \ref{physics crooks} and Proposition \ref{counterexample}), it is actually not an issue from an experimental point of view when the Markov chains correspond to thermodynamic processes. This is the case because of the way one is able to sample from a Boltzmann distribution given a thermodynamic system. One of the assumptions in Corollary \ref{physics crooks} is that $X_0$ should follow such a distribution for $E_0$. For this to be fulfilled, one needs to wait until the system relaxes to such a state. Because of that, one can think of any trajectory as having an additional point which was also sampled from the Boltzmann distribution for $E_0$. Thus, the assumption that $p_0$ is the equilibrium distribution of $M_1$ is always fulfilled and 
the experimental range of validity of Crooks' fluctuation theorem in our setup remains equal to the one in non-equilibrium thermodynamics \cite{crooks1998nonequilibrium,crooks2000path}.
In particular, the new constraint is fulfilled in previous experimental setups supporting the theorem (see, for example, \cite{collin2005verification} or \cite{saira2012test}).

	\paragraph{Example application: Crooks' theorem.}
	\label{appli crooks}
	Hysteresis is a well-known effect that takes place in some physical systems and refers to the difference in the system's behaviour when interacting with a series of environments compared to its response when facing the same conditions in reversed order \cite{jarzynski2011equalities}. The same idea also applies to decision-making systems. In fact, hysteresis has been reported in both simulations of decision-making systems \cite{grau2018non} as well as in biological decision-makers recorded experimentally \cite{turnham2012facilitation,hack2022thermodynamic}.  Given that it refers to the difference between decisions when the order in which the environments are presented is reversed, \eqref{physics Crooks eq II} and \eqref{eq: crooks physics} constitute quantitative measures of hysteresis. In particular, \cite{hack2022thermodynamic} has used this measure successfully to quantify hysteresis in human sensorimotor adaptation, where human learners had to solve a simple motor coordination task in a dynamic environment with changing visuomotor mappings. While a simple Markov model proved adequate to model sensorimotor adaptation, it should be noted that more complex learning scenarios involving long-term memory and abstraction would not be captured by such a simple model.

\paragraph{Detailed balance.}
Detailed balance is not required neither for Jarzynski's equality (Theorem \ref{Jarzynski's equality}) nor for the more general form of Crooks' fluctuation theorem we presented in Theorem \ref{crooks fluctuation thm}. It is, however, required in order to choose $\vect{Y}$ to be the time reversal of $\vect{X}$, which leads to Crooks' fluctuation theorem in Corollary \ref{physics crooks}. 

While the definition of detailed balance \eqref{det balance} we adopted here is standard in the Markov chain literature, there is some ambiguity regarding its use in thermodynamics, where it has, at least, two more meanings. It is used both for the weaker condition that the Boltzmann distribution $p_n$ is a stationary distribution of $M_n$ for $1 \leq n \leq N$ \cite{jarzynski1997equilibrium} and as a synonym of microscopic reversibility \cite{crooks2000path,crooks1998nonequilibrium}. Although we have shown microscopic reversibility and detailed balance are indeed equivalent under some conditions (see Lemma \ref{physical core crooks}), we have followed its definition in \cite{crooks2000path}, which is not the only one in the literature (see \cite{crooks2011thermodynamic,tolman1925principle,cohen2004note}).

Notice, what is called a stationary distribution in the literature on Markov chains is referred to as a \emph{nonequilibirum steady state} in thermodynamics \cite{zhang2012stochastic}. In order for it to be an \emph{equilibrium state}, it needs to fulfill detailed balance \eqref{det balance} with respect to the the transition matrix in question. Notice, also, detailed balance is not fulfilled is several application of non-equilibrium thermodynamics throughout physics \cite{tang2015work} and biology \cite{battle2016broken}.

%\paragraph{Experimental evidence supporting the fluctuation theorems.} In a decision-theoretic scenario, the first piece of evidence supporting the validity of the fluctuation theorems has been recently reported \cite{hack2022thermodynamic}. In the context of thermodynamics, experimental evidence supporting the fluctuation theorems has been reported in several contexts: unfolding and refolding processes involving RNA \cite{collin2005verification,liphardt2002equilibrium}, electronic transitions between electrodes manipulating a charge parameter \cite{saira2012test}, rotation of a macroscopic object inside a fluid surrounded by magnets where the current of a wire attached to the macroscopic object is manipulated \cite{douarche2005experimental}, and a trapped-ion system \cite{an2015experimental,smith2018verification}.

\paragraph{Continuous-time Markov chains.} Notice, in case we have a continuous-time Markov chain, work becomes an integral where the integrand for $\vect{X}$ at $\vect{x}$ and the one for $\vect{Y}$ at $\vect{x}^R$ differ, aside from the sign, in a single point. Thus, work is odd under time reversal and 
the assumption that $p_0=p_1$ can be dropped in both Theorem \ref{crooks fluctuation thm} and Corollary \ref{physics crooks}. However, the technical tools required to show Crooks' fluctuation theorem or Jarzynski's equality are technically more involved in the continuous-time case, as one can see in \cite{ge2007generalized}. 
%where the continuous-time version of both notice the method in Theorem \ref{Jarzynski's equality} can also be used to show it in the continuous-time case, although the proof is technically more involved. 

	\paragraph{Continuous state space.}
	In case the state space is continuous, the results can be derived in a similar fashion. What we ought to notice is that, in this scenario, the role of the transition matrices is played by the densities of the Markov kernels (see for example \cite{chib1995understanding}). These densities allow us to write conditions like detailed balance analogously to how we do it in the discrete case. In the case of Jarzynski's equality for a Markov chain on a continuous state space $\vect{X}$, one can see that the result follows like the one with a discrete state space. To convince ourselves, the only thing to take into account is the substitution of the sum in the expected value by the integral and that of the probability distribution by the density of the Markov kernel. Then, following the proof of Theorem \ref{Jarzynski's equality}, we can define a stochastic process $\vect{Y}$ whose density Markov kernels are defined through the stationary distributions and density Markov kernels of $\vect{X}$, in analogy to how we defined them in Theorem \ref{Jarzynski's equality}. The rest follow exactly in the same fashion. Crooks' fluctuation theorem requires a longer explanation but, essentially, follows from the same considerations.

\section{Conclusion}

In this paper, we have investigated the potential of thermodynamic fluctuation theorems to serve probabilistic laws of decision-making. In particular, we have derived two thermodynamic fluctuation theorems, Jarzynski's equality and Crooks' theorem, in the context of general Markov chains $\vect{X}$. We have started by defining several thermodynamic concepts for Markov chains and discussing how these definitions do not correspond in general to the ones used in thermodynamics. Right after, we have derived Jarzynski's equality in Theorem \ref{Jarzynski's equality} without any assumption involving the time reversal of $\vect{X}$. Thus, we have improved on the previous attempt to derive it in the context of decision-making \cite{grau2018non}, which was based on the physical conventions pioneered in \cite{crooks1998nonequilibrium}. Regarding Crooks' fluctuation theorem, we have shown in Theorem \ref{crooks fluctuation thm}, Corollary \ref{physics crooks}, and Proposition \ref{counterexample}, that, in our decision-theoretic setup, it requires the additional assumption that the initial distribution of $\vect{X}$ must be the stationary distribution of its first transition matrix. This results from the fact that our notion of work is inherent to Markov chains, which contrasts with the definition used in previous derivations \cite{crooks1998nonequilibrium,crooks2000path,crooks1999excursions}, where the work along the forward and backward paths is calculated in a different way for physical reasons. Instead, we calculate the work along the two paths in the same way, in order for the quantities involved in the final result to have an interpretation that is relevant for decision-making.

% Finally, we notice the additional requirement vanished for continuous-time Markov chains. 

%In this paper, we have dealt with the theoretical basis underlying both Jarzynski's equality and Crooks' fluctuation theorem for general Markov chains $\vect{X}$. Regarding the former, we have shown in Theorem \ref{Jarzynski's equality} that it can be derived without any assumption involving the time reversal of $\vect{X}$, in contrast to the otherwise similar approach in \cite{crooks1998nonequilibrium}. Regarding Crooks' fluctuation theorem, we have shown in Theorem \ref{crooks fluctuation thm}, Corollary \ref{physics crooks}, and Proposition \ref{counterexample}, that it requires the additional assumption that the initial distribution of $\vect{X}$ must be the stationary distribution of its first transition matrix, which is absent in the literature on the topic. However, although important from a conceptual perspective, the additional hypothesis does not affect previous experimental evidence supporting Crooks' fluctuation theorem, as it is always fulfilled given the way in which one can, in practice, prepare a thermodynamic system in the Boltzmann distribution.

\newpage

\begin{appendix}
\section{Appendix}
\label{appendix}

Here, we show two results relating a Markov chain $\vect{X}$ to the Markov chain $\vect{Y}$ defined from $\vect{X}$ as in Theorem \ref{Jarzynski's equality}. First, in Lemma \ref{X irred implies Y irred}, we relate the irreducibility of transition matrices for $\vect{X}$ with that of $\vect{Y}$.

\begin{lem}
\label{X irred implies Y irred}
If $\vect{X}=(X_n)_{n=0}^N$ is a Markov chain on a finite state space $S$ whose transition matrices are irreducible, then the transition matrices of the Markov chain $\vect{Y}=(Y_n)_{n=0}^N$ defined in Theorem \ref{Jarzynski's equality} are irreducible. 
\end{lem}

\begin{proof}
We denote by $(M_n)_{n=1}^N$ the transition matrices of $\vect{X}$ and by $(\widehat{M}_n)_{n=1}^N$ the ones of $\vect{Y}$. Consider $x,y \in S$ and $\widehat{M}_n$ for some $1 \leq n \leq N$. Since $M_{N+1-n}$ is irreducible, there exist $m \in \mathbb{N}$ and $x_0,x_1,..,x_m$ where $x_0=y$ and $x_m=x$ such that 
\begin{equation*}
\label{irreducible}
\prob(X_m=x_m,..,X_1=x_1|X_0=x_0) = (M_{N+1-n})_{x_m,x_{m-1}}..(M_{N+1-n})_{x_1,x_0} >0. 
\end{equation*}
Hence, we have
\begin{equation*}
\begin{split}
    &\prob(Y_m=x_0,..,Y_1=x_{m-1}|Y_0=x_m) \stackrel{(i)}{=} (\widehat{M}_{n})_{x_0,x_{1}}..(\widehat{M}_{n})_{x_{m-1},x_m} \\
    &\stackrel{(ii)}{=} (M_{N+1-n})_{x_m,x_{m-1}}..(M_{N+1-n})_{x_1,x_0} \frac{p_{N+1-n}(x_{m-1})}{p_{N+1-n}(x_m)}..\frac{p_{N+1-n}(x_0)}{p_{N+1-n}(x_1)} \\
    &\stackrel{(iii)}{>}0
    \end{split}
\end{equation*}
where we applied the Markov property of $\vect{Y}$ in $(i)$, \eqref{def rev} in $(ii)$ and \eqref{irreducible} plus the fact, as $M_{N+1-n}$ is irreducible, $p_{N+1-n}(x)>0$ $\forall x \in S$ by Lemma \ref{irred prop} in $(iii)$.
\end{proof}

The connection between the irreducibility of the transition matrices of $\vect{X}$ and $\vect{Y}$ in Lemma \ref{X irred implies Y irred}
%Hence, whenever the transition matrices of $\vect{X}$ are irreducible, energy (thus work) is well-defined up to a constant for $\vect{Y}$ by Lemma \ref{irred prop} and \ref{X irred implies Y irred}. In fact,
results in a relation between their energy functions, which we prove in Lemma \ref{energy rela}.

\begin{lem}
\label{energy rela}
If $\vect{X} = (X_n)_{n=0}^N$ is a Markov chain on a finite state space $S$ whose initial distribution $p_0$ has non-zero entries and whose transition matrices $(M_n)_{n=1}^N$ are irreducible, $\vect{Y}=(Y_n)_{n=0}^N$ is the Markov chain defined in Theorem \ref{Jarzynski's equality}, $\vect{E}= (E_n)_{n=0}^N$ is a family of energy functions of $\vect{X}$ and $\vect{\widehat{E}} = (\widehat{E}_n)_{n=0}^N$ is a family of energy functions of $\vect{Y}$, then there exist constants $(k_n)_{n=0}^N$ such that $\widehat{E}_{0} = E_{N} + k_{0}$ and 
\begin{equation}
\label{const diff}
    \widehat{E}_{n+1} = E_{N-n} + k_{n+1}
\end{equation} 
for $0 \leq n \leq N-1$.
\end{lem}

\begin{proof}
Since all transition matrices of $\vect{X}$ are irreducible and the same holds for $\vect{Y}$ by Lemma \ref{X irred implies Y irred}, we can apply Lemma \ref{irred prop} to $\vect{Y}$ and get each of its transition matrices has a unique stationary distribution composed of non-zero entries. Thus, energy is well-defined, up to a constant, for $\vect{Y}$. Since $Y_0$ follows $p_N$ by definition, we have automatically there exists a constant $k_0$ such that $\widehat{E}_{0} = E_{N} + k_{0}$. To get \eqref{const diff}, notice we have $\widehat{M}_{n+1} p_{N-n} = p_{N-n}$ for $0 \leq n \leq N-1$ by definition of $(\widehat{M}_n)_{n=1}^N$:
\begin{equation}
\label{equi for reverse}
\begin{split}
    (\widehat{M}_{n+1} p_{N-n})(y) &= \sum_{x \in S} (\widehat{M}_{n+1})_{y x} p_{N-n}(x) = \sum_{x \in S} (M_{N-n})_{x y} p_{N-n}(y)\\ &=p_{N-n}(y),
    \end{split}
\end{equation}
where we applied \eqref{def rev} in the second equality and the fact $M_{N-n}$ is a stochastic matrix in the third. Thus, for $0 \leq n \leq N-1$, $p_{N-n}$ is the unique stationary distribution of $\widehat{M}_{n+1}$ and there exists a constant $k_{n+1}$ such that $\widehat{E}_{n+1} = E_{N-n} + k_{n+1}$.
\end{proof}

\end{appendix}

\newpage

\bibliographystyle{plain}
\bibliography{v5.bib}

\end{document}